\documentclass[a4paper,UKenglish]{lipics-v2016}
 
\usepackage{microtype}

\usepackage{microtype}
\usepackage{graphicx}
\usepackage{graphics}
\usepackage{epsfig}
\usepackage{epstopdf}
\usepackage{amsmath}
\usepackage{latexsym}
\usepackage{wrapfig}
\usepackage{multirow}
\usepackage{amsfonts}
\usepackage{cite}
\usepackage{amssymb}
\usepackage{caption}
\usepackage{algorithm}
\usepackage{algorithmic}
\usepackage{enumerate}
\usepackage{diagbox}

\makeatletter
\newcommand{\rmnum}[1]{\romannumeral #1}
\newcommand{\Rmnum}[1]{\expandafter\@slowromancap\romannumeral #1@}
\makeatother

\title{On Geometric Prototype and Applications}
\titlerunning{On Geometric Prototype} 

\author[1]{Hu Ding}
\author[2]{Manni Liu}
\affil[1]{Michigan State University, East Lansing, USA\\
  \texttt{huding@msu.edu}}
\affil[2]{Michigan State University, East Lansing, USA\\
  \texttt{liumanni@msu.edu}}
\authorrunning{H. Ding and M. Liu} 

\Copyright{H. Ding and M. Liu}

\subjclass{F.2.2 Nonnumerical Algorithms and Problems - Geometrical
problems and computations}
\keywords{prototype, hardness, core-set, Wasserstein barycenter, ensemble clustering}

\EventEditors{}
\EventNoEds{2}
\EventLongTitle{26th Annual European Symposium on Algorithms (ESA 2018)}
\EventShortTitle{ESA 2018}
\EventAcronym{ESA}
\EventYear{2018}
\EventDate{}
\EventLocation{}
\EventLogo{}
\SeriesVolume{}
\ArticleNo{}

\begin{document}

\maketitle

\begin{abstract}
In this paper, we propose to study a new geometric optimization problem called ``geometric prototype'' in Euclidean space. Given a set of patterns, where each pattern is represented by a (weighted or unweighted) point set, the geometric prototype can be viewed as the ``mean pattern'' minimizing the total matching cost to them. As a general model, the problem finds many applications in the areas like machine learning, data mining, computer vision, etc. The dimensionality could be either constant or high, depending on the applications. To our best knowledge, the general geometric prototype problem has yet to be seriously considered by the theory community. To bridge the gap between theory and practice, we first show that a small core-set can be obtained to substantially reduce the data size. Consequently, any existing heuristic or algorithm can run on the core-set to achieve a great improvement on the efficiency. As a new application of core-set, it needs to tackle a couple of challenges particularly in theory. Finally, we test our method on both 2D image and high dimensional clustering datasets; the experimental results remain stable even if we run the algorithms on the core-sets much smaller than the original datasets, while the running times are reduced significantly.

%
%
%
 \end{abstract}

\section{Introduction}
\label{sec-intro}

%

Given a set of points in Euclidean space, we can easily use their mean point to represent them. However, if they are replaced by a set of point sets where each point set denotes a ``pattern'', the problem of finding their representation will be much more challenging. We call it ``Geometric Prototype'' problem. Before introducing its formal definition, we need to define the matching cost between two patterns first.

\begin{definition}[$\mathcal{M}(A, B)$]
\label{def-match}
Given two point sets $A=\{a_1, a_2, \cdots, a_k\}$ and $B=\{b_1, b_2, \cdots, b_k\}$ in $\mathbb{R}^d$, 
\begin{eqnarray}
\mathcal{M}(A,B)=\min_{\pi\in\Pi}\sum^k_{j=1}||a_j-b_{\pi(j)}||^2 \label{for-match}
\end{eqnarray}
where $\Pi$ contains all the possible permutations of $\{1, 2, \cdots, k\}$. 
\end{definition}

$\mathcal{M}(A, B)$ is in fact the problem of geometric matching which can be optimally solved by Hungarian algorithm~\cite{cormenintroduction}.
 When the dimensionality is constant, a number of efficient approximation algorithms have been developed in past years (see more discussion in Section~\ref{sec-contribution}).

\begin{definition} [Geometric Prototype]
\label{def-gp}
Given a set of point sets $\mathbb{P}=\{P_1, P_2, \cdots, P_n\}$ with each $P_i$ containing $k$ points $\{p^i_1, p^i_2, \cdots, p^i_k\}\subset\mathbb{R}^d$, the geometric prototype is a new point set $g(\mathbb{P})$ having $k$ points such that 
\begin{eqnarray}
\sum^n_{i=1}\mathcal{M}(P_i, g(\mathbb{P})) \label{for-def-gp}
\end{eqnarray}
is minimized. Note $g(\mathbb{P})$ is not necessarily from $\mathbb{P}$. Also, any $k$-point set achieving at most $c$ times the minimum value of (\ref{for-def-gp}) is called a $c$-approximation with $\forall c\geq 1$.
\end{definition}
\begin{remark}
It is easy to see that when $k=1$, the geometric prototype is simply the mean point. Actually, the problem of geometric prototype can be viewed as a ``chromatic $k$-means clustering''. The $kn$ points of $\cup^n_{i=1}P_i$ form $k$ clusters where the $k$ points of each $P_i$ should be assigned to the $k$ clusters separately; to minimize the objective function~(\ref{for-def-gp}), the $k$ points of $g(\mathbb{P})$ should be the mean points of the resulting clusters.
\end{remark}

In Definition~\ref{def-gp}, the dimension $d$ could be either constant or high depending on the applications, and $n$ usually is large ($k$ could be not constant, but often much smaller than $n$ in the applications). To our best knowledge, the general geometric prototype problem has never been systematically studied in the area of computational geometry (except some special cases; see Section~\ref{sec-contribution}), but finds many real-world applications recently. Below, we introduce two important applications in low and high dimension, respectively.

\vspace{0.05in}

\textbf{(1) Wasserstein Barycenter.} Given a large set of images, finding their average yields several benefits in practice. For example, if all the images are taken from the same object but have certain extents of noise, their average image could serve as a robust pattern to represent them; also, this is an efficient way to compress large image datasets. In computer vision, Earth Mover's Distance (EMD)~\cite{rubner2000earth} is widely used to measure the difference between two images; the average image minimizing the total EMDs to all the images is defined as the {\em Wasserstein Barycenter}~\cite{cuturi2014fast,baum2015wasserstein,gramfort2015fast,ye2017fast,DBLP:journals/siamsc/BenamouCCNP15}.  In addition, Ding and Xu~\cite{ding2013k,ding2014finding} considered the case allowing rigid/affine transformations for each image. Wasserstein Barycenter can also be applied to Bayesian inference~\cite{staib2017parallel}. Note that the geometric prototype defined above is not exactly equivalent to Wasserstein Barycenter, because the latter one requires each point having a non-negative weight and EMD is to minimize the max flow cost; however, the techniques proposed in this paper can be easily extended to handle EMD and we will discuss it later.

\vspace{0.05in}

\textbf{(2) Ensemble Clustering.} Given a number of different clustering solutions for the same set of items, the problem about finding a unified clustering solution minimizing the total differences to them is called {\em ensemble clustering}~\cite{GA13}. This problem has attracted a great deal of attention, especially for the applications in big data and crowdsourcing~\cite{SG02,GHL,GLF,SMP}. For example, due to the proliferation of networked sensing systems, we can use a large number of sensors to record the same environment and each sensor can generate an individual clustering for the same set of objects. However, most of existing approaches rely on algebraic or graphic models and need to solve  complicated optimizations with high complexities (such as semi-definite programming~\cite{SMP}). 

Recently, Ding et al.~\cite{DBLP:conf/mobihoc/DingSX16} present a novel high dimensional geometric model for the problem of ensemble clustering: suppose there are $d$ items and each clustering solution has $k$ clusters on these items (if less than $k$, we can add some dummy empty clusters); then, each single cluster is mapped to a binary vector in $\mathbb{R}^d$ where each dimension indicates the membership of an individual item (see Figure~\ref{fig-map}); so each clustering solution is mapped to a $k$-point set in $\mathbb{R}^d$; the size of the symmetric difference between two clusters is equal to their squared distance in $\mathbb{R}^d$, and thus the difference between two clustering solutions is always equal to half of their matching cost (Definition~\ref{def-match}) in Euclidean space. Therefore, finding the final clustering solution minimizing the total differences to the given solutions is equivalent to computing the geometric prototype of the resulting $k$-point sets in $\mathbb{R}^d$. Please find more details in~\cite{DBLP:conf/mobihoc/DingSX16}. Note that the obtained geometric prototype may result in fractional clustering memberships, because the points of the geometric prototype are not necessarily binary vectors. So the approximation result in~\cite{DBLP:conf/mobihoc/DingSX16} does not violate the APX-hardness for strict ensemble/consensus clustering~\cite{bonizzoni2008approximation}. Actually fractional clustering memberships are acceptable and make sense in practice; for instance, we may claim that one object belongs to class 1, 2, and 3 with probabilities of $70\%$, $20\%$, and $10\%$,  respectively.
\begin{figure}[h]
\centering
\includegraphics[height=1in]{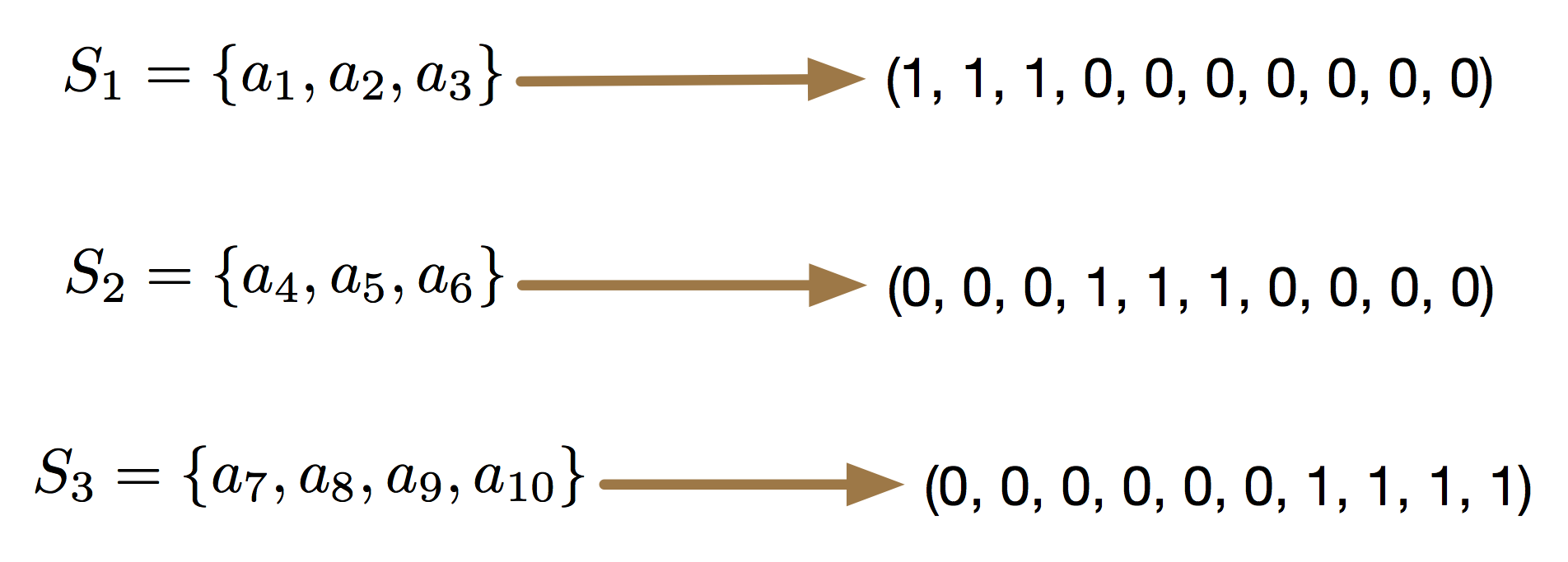}
\vspace{-0.1in}
\caption{$d=10$ and $k=3$. The three clusters are mapped to $3$ binary vectors in $\mathbb{R}^{10}$.}
\label{fig-map}
\end{figure}
\vspace{-0.2in}

\subsection{Our Main Contributions and Related Work}
\label{sec-contribution}

Due to the non-convex nature of geometric prototype problem, most of the aforementioned approaches for Wasserstein barycenter~\cite{cuturi2014fast,baum2015wasserstein,gramfort2015fast,ye2017fast,DBLP:journals/siamsc/BenamouCCNP15} and large-scale ensemble clustering~\cite{DBLP:conf/mobihoc/DingSX16} are iterative algorithms, such as alternating minimization and Alternating Direction Method of Multipliers (ADMM)~\cite{boyd2011distributed}, which can converge to some local optimums. Those approaches could be very slow for large datasets, because they may run many rounds and each round usually needs to conduct some complicated update or optimization. This is also the main motivation of our work, that is, replacing the original large input by a small core-set to speed up the computation of existing algorithms.

In this paper, our contribution is twofold in the aspects of theory and applications. In theory, we show that a small core-set can be obtained for the problem of geometric prototype. 
More importantly, our core-set is independent of any geometric prototype algorithm; namely,  we can run any available algorithm as a black box on the core-set, instead of the original instance $\mathbb{P}$, to achieve a similar result. Although core-set has been extensively studied for many applications before~\cite{DBLP:journals/corr/Phillips16,agarwal2005geometric}, we still need to tackle several significant challenges when constructing the core-set for geometric prototype. In practice, we test our method for solving the applications Wasserstein barycenter and ensemble clustering. The experiment shows that running the existing algorithms on core-sets can achieve almost the same results while the running times are substantially reduced.

\textbf{Related work.} The general geometric prototype problem has yet to be seriously considered by the theory community (to our best knowledge), however, some special cases were studied before. Based on the remark below Definition~\ref{def-gp}, we know that finding the geometric prototype is also a chromatic clustering problem. Motivated by the application of managing traffic flows, Arkin et al.~\cite{arkin2015bichromatic} studied a variety of chromatic $2$-center clustering in 2D and gave both exact and approximation solutions. In addition, Ding and Xu~\cite{Din11,DX15} studied chromatic clustering in high dimension; however, their method assumes that $k$ is constant and thus it is unable to be extended to our general geometric prototype problem. 

Computing the geometric matching $\mathcal{M}(A,B)$ is a sub-problem of geometric prototype. Besides Hungarian algorithm~\cite{cormenintroduction}, the computational geometry community has extensively studied its approximation algorithms for the case in constant dimension~\cite{DBLP:conf/compgeom/AgarwalFPVX17,DBLP:conf/compgeom/AgarwalV04,DBLP:conf/soda/SharathkumarA12,DBLP:conf/stoc/SharathkumarA12,DBLP:conf/stoc/AndoniNOY14}, and some of them can achieve nearly linear running time.  
 
The rest of the paper is organized as follows. We first introduce some basic results and useful tools in Section~\ref{sec-pre}. Then we show our core-set construction and analysis in Section~\ref{sec-coreset}. Finally, we implement our algorithm and  test it on multiple datasets in Section~\ref{sec-exp}.

\section{Preliminaries}
\label{sec-pre}

\textbf{The hardness.} Actually, we are able to show that finding the optimal geometric prototype of a given instance is NP-hard and has no FPTAS even if $k=2$ in high dimensional space, unless P=NP. Our proof makes use of the construction by Dasgupta for the NP-hardness proof of $2$-means clustering problem in high dimension~\cite{D08}. See Section~\ref{sec-hard} in Appendix for details. Moreover, we leave the hardness for the low dimensional case of geometric prototype as an open problem in future work.

The following lemma, which can be easily obtained via Definition~\ref{def-match}, is repeatedly used in our analysis (see the proof in Section~\ref{sec-plem-tri} of Appendix).

\begin{lemma}
\label{lem-tri}
Given three $k$-point sets $A$, $B$, and $C$ in $\mathbb{R}^d$, 
\begin{eqnarray}
\mathcal{M}(A,B)\leq 2\mathcal{M}(A,C)+2\mathcal{M}(C,B).
\end{eqnarray}
\end{lemma}

Using Markov inequality and Lemma~\ref{lem-tri}, Ding et al.~\cite{DBLP:conf/mobihoc/DingSX16} showed that a constant approximation can be achieved with constant probability.

\begin{theorem}[\cite{DBLP:conf/mobihoc/DingSX16}]
\label{the-constant}
Let $\alpha>1$. Given an instance $\mathbb{P}$ of geometric prototype problem, if we randomly pick a point set $P_{i_0}$ from $\mathbb{P}$, then with probability at least $1-\frac{1}{\alpha}$,  $\mathcal{M}(P_{i_0}, g(\mathbb{P}))$ is no larger than $\frac{\alpha}{n}\sum^n_{i=1}\mathcal{M}(P_i, g(\mathbb{P}))$ and $P_{i_0}$ yields a $(2\alpha+2)$-approximation.
\end{theorem}
\begin{remark}
To boost the success probability, we can try multiple times and select the one yielding the lowest objective value. For example, if we try $t$ times, the success probability will be $1-\frac{1}{\alpha^t}$. 
\end{remark}

According to Theorem~\ref{the-constant}, the selected $P_{i_0}$ could serve as a good initialization for the geometric prototype. To further improve the approximation ratio, the algorithm in~\cite{DBLP:conf/mobihoc/DingSX16} adopts a simple alternating minimization procedure, i.e., alternatively updating the prototype and matchings round by round. 
The main drawback of this algorithm is that it needs to repeatedly compute the matchings between the prototype and all the given point sets in each round, and thus the running time is high especially when some or all of $n$, $k$, and $d$ are large (as discussed at the beginning of Section~\ref{sec-contribution}).


In addition, we are able to apply the well known Johnson-Lindenstrauss (JL) lemma~\cite{achlioptas2003database} to reduce the dimensionality before running the algorithm; also, the obtained geometric prototype in the lower dimension can be efficiently mapped back to the original space~\cite{DBLP:conf/mobihoc/DingSX16}. 


\begin{theorem}[\cite{DBLP:conf/mobihoc/DingSX16}]
\label{the-jl}
Let $0<\epsilon<1$ and $c\geq 1$. Suppose we randomly project a given instance $\mathbb{P}$ of geometric prototype problem from $\mathbb{R}^d$ to $\mathbb{R}^{O(\log (nk)/\epsilon^2)}$ and obtain a new instance $\mathbb{P}'$ in the lower dimension. Then, with high probability, we can convert any $c$-approximation for $\mathbb{P}'$ to a $c(\frac{1+\epsilon}{1-\epsilon})^2$-approximation for $\mathbb{P}$ in $\mathbb{R}^d$, in $O(nkd)$ time.
\end{theorem}

%
%
%

The following lemma is a key tool in our analysis. In fact, it can be viewed as an interesting supplement of Lemma~\ref{lem-tri}.

\begin{lemma}
\label{lem-tri2}
Let $A$, $B$, and $C$ be three $k$-point sets in $\mathbb{R}^d$. Then for any $\epsilon>0$, 
\begin{eqnarray}
\Big|\mathcal{M}(A, B)-\mathcal{M}(A, C)\Big|\leq (1+\frac{1}{\epsilon})\mathcal{M}(B,C)+\epsilon\mathcal{M}(A,B)\label{for-tri2}
\end{eqnarray}
\end{lemma}
\begin{proof}
Let $A=\{a_1, a_2, \cdots, a_k\}$, $B=\{b_1, b_2, \cdots, b_k\}$, and $C=\{c_1, c_2, \cdots, c_k\}$. 

First, we consider the case $\mathcal{M}(A, B)\geq\mathcal{M}(A, C)$. W.l.o.g, we assume that the induced permutations of $\mathcal{M}(A, C)$ and $\mathcal{M}(B,C)$ are both $\pi(j)=j$ for $1\leq j\leq k$ (since these two permutations are independent with each other). Then we have
\begin{eqnarray}
\Big|\mathcal{M}(A, B)-\mathcal{M}(A, C)\Big|&=&\mathcal{M}(A, B)-\mathcal{M}(A, C)\nonumber\\
&\leq&\sum^k_{j=1}||a_j-b_j||^2-\mathcal{M}(A, C)\nonumber\\
&=&\sum^k_{j=1}\Big(||a_j-b_j||^2-||a_j-c_j||^2\Big)\nonumber\\
&\leq&\sum^k_{j=1}\Big(||a_j-b_j||+||a_j-c_j||\Big)||c_j-b_j||\nonumber\\
&\leq& \sum^k_{j=1}\Big(||c_j-b_j||+2||a_j-c_j||\Big)||c_j-b_j||\nonumber\\
&=&\sum^k_{j=1}\Big(||c_j-b_j||^2+2||a_j-c_j||\cdot||c_j-b_j||\Big) \nonumber\\
&\leq&\sum^k_{j=1}\Big(||c_j-b_j||^2+\epsilon||a_j-c_j||^2+\frac{1}{\epsilon}||c_j-b_j||^2\Big)\nonumber\\
&=&(1+\frac{1}{\epsilon})\mathcal{M}(B, C)+\epsilon\mathcal{M}(A, C)\label{for-tri2-1}
\end{eqnarray}
via repeatedly applying triangle inequality. Because we assume $\mathcal{M}(A, B)\geq\mathcal{M}(A, C)$ for this case, (\ref{for-tri2-1}) implies that (\ref{for-tri2}) holds.

For the other case $\mathcal{M}(A, B)<\mathcal{M}(A, C)$, we directly have
\begin{eqnarray}
\Big|\mathcal{M}(A, B)-\mathcal{M}(A, C)\Big|\leq (1+\frac{1}{\epsilon})\mathcal{M}(B,C)+\epsilon\mathcal{M}(A,B) \label{for-tri2-2}
\end{eqnarray}
by exchanging the roles of $B$ and $C$ in (\ref{for-tri2-1}). 
%
\end{proof}

\section{Core-set for Reducing the Data Size}
\label{sec-coreset}
Langberg and Schulman~\cite{langberg2010universal} introduced a framework of core-set (it was called ``$\epsilon$-approximator'' in their paper) to compress data for several geometric shape fitting problems; further, Feldman and Langberg~\cite{DBLP:conf/stoc/FeldmanL11}  improved the core-set size for a large class of clustering problems. Here, we consider to construct a core-set of the instance $\mathbb{P}$ so as to reduce the data size and running time. Formally, our objective is to find a small sample $\mathbb{S}\subset\mathbb{P}$ and assign a weight $w_l$ for each $P_l\in \mathbb{S}$, such that for any $k$-point set $Q\subset\mathbb{R}^d$, 
\begin{eqnarray}
\Big|\sum_{P_l\in\mathbb{P}}\mathcal{M}(P_l, Q)-\sum_{P_l\in\mathbb{S}}w_l\mathcal{M}(P_l, Q)\Big|\leq O(\epsilon) \sum_{P_l\in\mathbb{P}}\mathcal{M}(P_l, Q) \label{for-coreset}
\end{eqnarray}
with certain probability and small enough $\epsilon>0$. Moreover, we want to keep each weight $w_l$ to be non-negative so as to easily run any existing algorithm or heuristic on the core-set.   
  
%
%
%
%
%
%
%
%
%
Unfortunately, we cannot directly apply the existing ideas to the problem of geometric prototype, because the points from $\cup^n_{i=1}P_i$ are not independent with each other (due to the matching constraint in Definition~\ref{def-match}; also see our remark below Definition~\ref{def-gp}) and it would be much more challenging to build the connection between the sampled core-set and $\mathbb{P}$. Instead, we regard each $P_i$ as an ``abstract point'' and compute a core-set on these $n$ abstract points. Though these abstract points can form some metric space with the matching costs being their pairwise (squared) distances, it is still quite different to metric clustering studied by~\cite{langberg2010universal,DBLP:conf/stoc/FeldmanL11,chen2009coresets}, since the prototype $g(\mathbb{P})$ is not necessarily from $\mathbb{P}$ and could appear anywhere in the Euclidean space. 

Conceptually, the core-set construction is a random sampling process: first, compute an upper bound on the sensitivity $\sigma_{\mathbb{P}}(P_i)$ of each $P_i$ (we will formally define the sensitivity later); then take a sample from $\mathbb{P}$ with  probabilities proportional to $\sigma_{\mathbb{P}}(P_i)$ to form the core-set. To implement this construction,  we have to develop new ideas for resolving the following two issues. \textbf{(\Rmnum{1})}~How to compute $\sigma_{\mathbb{P}}(P_i)$, or its upper bound, so as to generate the probability distribution for sampling. \textbf{(\Rmnum{2})}~What about the sample size. We consider these two issues in Section~\ref{sec-issue1} and \ref{sec-issue2}, respectively. The final result for core-set construction of geometric prototype is presented in Theorem~\ref{the-coreset}. We also discuss some extensions on other metrics (e.g., $l_1$ norm and earth mover's distance) and the time complexity in Section~\ref{sec-coreext} and~\ref{sec-time}, respectively.

\subsection{Solving Issue \Rmnum{1}}
\label{sec-issue1}
Following~\cite{langberg2010universal}, the sensitivity of each $P_i\in \mathbb{P}$ is defined as follows:
\begin{eqnarray}
\sigma_{\mathbb{P}}(P_i)={sup}_{Q}\frac{\mathcal{M}(P_i, Q)}{\sum_{P_l\in\mathbb{P}}\mathcal{M}(P_l, Q)}\label{for-sen1}
\end{eqnarray}
where $Q$ is restricted to be $k$-point set in $\mathbb{R}^d$. Intuitively, the sensitivity measures the importance of each $P_i$ among all the patterns of $\mathbb{P}$. Directly obtaining the value of $\sigma_{\mathbb{P}}(P_i)$ could be challenging and also needless, thus we often turn to compute an upper bound for it. 

Recall that $g(\mathbb{P})$ is the optimal geometric prototype of $\mathbb{P}$, and we let $\Delta=\sum_{P_l\in\mathbb{P}}\mathcal{M}(P_l, g(\mathbb{P}))$ for convenience.

\begin{lemma}
\label{lem-issue1}
For any $P_i\in\mathbb{P}$, $\sigma_{\mathbb{P}}(P_i)\leq \frac{2\mathcal{M}(P_i, g(\mathbb{P}))}{\Delta}+\frac{16}{n}$.
\end{lemma}
\begin{proof}
First, we consider $\frac{\mathcal{M}(P_i, Q)}{\sum_{P_l\in\mathbb{P}}\mathcal{M}(P_l, Q)}$ with a fixed $Q$ in (\ref{for-sen1}). Through Lemma~\ref{lem-tri}, we know that the numerator $\mathcal{M}(P_i, Q)$ is bounded by $2\mathcal{M}(P_i,g(\mathbb{P}))+2\mathcal{M}(g(\mathbb{P}),Q)$. Then, we consider two cases: \textbf{(1)} $\mathcal{M}(g(\mathbb{P}),Q)\leq \frac{8}{n}\Delta$ and \textbf{(2)} $\mathcal{M}(g(\mathbb{P}),Q)> \frac{8}{n}\Delta$.

Since $\Delta\leq \sum_{P_l\in\mathbb{P}}\mathcal{M}(P_l, Q)$, we directly have 
\begin{eqnarray}
\frac{\mathcal{M}(P_i, Q)}{\sum_{P_l\in\mathbb{P}}\mathcal{M}(P_l, Q)}&\leq& \frac{2\mathcal{M}(P_i,g(\mathbb{P}))+2\mathcal{M}(g(\mathbb{P}),Q)}{\sum_{P_l\in\mathbb{P}}\mathcal{M}(P_l, Q)}\nonumber\\
&\leq& \frac{2\mathcal{M}(P_i,g(\mathbb{P}))+\frac{16}{n}\Delta}{\Delta}=\frac{2\mathcal{M}(P_i, g(\mathbb{P}))}{\Delta}+\frac{16}{n}
\end{eqnarray}
for case (1).

Now, we assume that case (2) is true. Denote by $\mathbb{P}'$ the set $\{P_l\in \mathbb{P}\mid \mathcal{M}(P_l, g(\mathbb{P}))\leq \frac{2}{n}\Delta\}$, and Markov inequality implies $|\mathbb{P}'|\geq \frac{n}{2}$. Applying Lemma~\ref{lem-tri} again, we have
\begin{eqnarray}
\sum_{P_l\in\mathbb{P}}\mathcal{M}(P_l, Q)&\geq& \sum_{P_l\in\mathbb{P}'}\mathcal{M}(P_l, Q)\geq \sum_{P_l\in\mathbb{P}'}\Big(\frac{1}{2}\mathcal{M}(g(\mathbb{P}),Q)-\mathcal{M}(g(\mathbb{P}),P_l)\Big)\nonumber\\
&\geq& \sum_{P_l\in\mathbb{P}'}\Big(\frac{1}{2}\mathcal{M}(g(\mathbb{P}),Q)-\frac{2}{n}\Delta\Big)\geq\frac{n}{2}\Big(\frac{1}{2}\mathcal{M}(g(\mathbb{P}),Q)-\frac{2}{n}\Delta\Big)\nonumber\\
&=&\frac{n}{4}\mathcal{M}(g(\mathbb{P}),Q)-\Delta.\label{for-sen2}
\end{eqnarray}
As a consequence,
\begin{eqnarray}
\frac{\mathcal{M}(P_i, Q)}{\sum_{P_l\in\mathbb{P}}\mathcal{M}(P_l, Q)}\leq \frac{2\mathcal{M}(P_i,g(\mathbb{P}))+2\mathcal{M}(g(\mathbb{P}),Q)}{\frac{n}{4}\mathcal{M}(g(\mathbb{P}),Q)-\Delta}. \label{for-sen3}
\end{eqnarray}
Since both $\mathcal{M}(P_i,g(\mathbb{P}))$ and $\Delta$ are independent of $Q$, the right-hand side of (\ref{for-sen3}) can be viewed as a function on $\mathcal{M}(g(\mathbb{P}),Q)$. Through a simple calculation and the assumption of case (2) (i.e., $\mathcal{M}(g(\mathbb{P}),Q)> \frac{8}{n}\Delta$), we know that it is always less than $\frac{2\mathcal{M}(P_i, g(\mathbb{P}))}{\Delta}+\frac{16}{n}$.

Overall, we have $\sigma_{\mathbb{P}}(P_i)\leq \frac{2\mathcal{M}(P_i, g(\mathbb{P}))}{\Delta}+\frac{16}{n}$ for both cases.
\end{proof}

However, only Lemma~\ref{lem-issue1} is not enough to compute the upper bound for $\sigma_{\mathbb{P}}(P_i)$, because neither $\mathcal{M}(P_i, g(\mathbb{P}))$ nor $\Delta$ is known. Therefore, we need to compute an approximation to replace the upper bound given by Lemma~\ref{lem-issue1}. 


\begin{lemma}
\label{lem-issue12}
Suppose $P_{i_0}$ is randomly picked from $\mathbb{P}$, and let $\tilde{\Delta}=\sum_{P_l\in\mathbb{P}}\mathcal{M}(P_l, P_{i_0})$ and $\alpha>1$. Then with probability $1-\frac{1}{\alpha}$, for all $1\leq i\leq n$, $\sigma_{\mathbb{P}}(P_i)\leq 8(\alpha+1)\frac{\mathcal{M}(P_i, P_{i_0})}{\tilde{\Delta}}+\frac{4\alpha+16}{n}$.
\end{lemma} 
\begin{proof}
According to Theorem~\ref{the-constant}, we know that $\mathcal{M}(P_{i_0}, g(\mathbb{P}))\leq \frac{\alpha}{n}\Delta$ and $\tilde{\Delta}\leq 2(\alpha+1)\Delta$ with probability at least $1-\frac{1}{\alpha}$. Then we have 
\begin{eqnarray}
\sigma_{\mathbb{P}}(P_i)&\leq& \frac{2\mathcal{M}(P_i, g(\mathbb{P}))}{\Delta}+\frac{16}{n}\leq\frac{4\mathcal{M}(P_i, P_{i_0})+4\mathcal{M}(P_{i_0}, g(\mathbb{P}))}{\Delta}+\frac{16}{n}\nonumber\\
&\leq& \frac{4\mathcal{M}(P_i, P_{i_0})}{\frac{1}{2(\alpha+1)}\tilde{\Delta}}+\frac{4\mathcal{M}(P_{i_0}, g(\mathbb{P}))}{\Delta}+\frac{16}{n}\leq 8(\alpha+1)\frac{\mathcal{M}(P_i, P_{i_0})}{\tilde{\Delta}}+\frac{4\alpha+16}{n},
\end{eqnarray}
where the first inequality comes from Lemma~\ref{lem-issue1}. So the proof is completed.
\end{proof}

Lemma~\ref{lem-issue12} indicates that once $P_{i_0}$ is selected, we can obtain an upper bound for each $\sigma_{\mathbb{P}}(P_i)$ by computing the values $\mathcal{M}(P_i, P_{i_0})$ and $\tilde{\Delta}$.

\subsection{Solving Issue \Rmnum{2}}
\label{sec-issue2}

Let $t_{\mathbb{P}}(P_i)$ and $T$ denote the obtained upper bound of $\sigma_{\mathbb{P}}(P_i)$ from Lemma~\ref{lem-issue12} and their sum, respectively. It is easy to know that $T=\sum_{P_i\in\mathbb{P}}t_{\mathbb{P}}(P_i)\leq 8(\alpha+1)+4\alpha+16$ which is constant if $\alpha$ is constant. For the sake of simplicity, we always assume $T=O(1)$ in our analysis below. 
%
%

We have the following theorem from~\cite{langberg2010universal,DBLP:conf/soda/VaradarajanX12} (we slightly modify their statements to fit our problem better).

\begin{theorem}[\cite{langberg2010universal,DBLP:conf/soda/VaradarajanX12}]
\label{the-lang}
Let $Q$ be any fixed $k$-point set in $\mathbb{R}^d$. \textbf{\rmnum{1}.} If we take a sample $P_i$ from $\mathbb{P}$ according to the distribution $\frac{t_{\mathbb{P}}(P_i)}{T}$, the expectation of $\frac{T}{t_{\mathbb{P}}(P_i)}\mathcal{M}(P_i, Q)$ is $\sum_{P_l\in\mathbb{P}}\mathcal{M}(P_l, Q)$. \textbf{\rmnum{2}.} If we take a sample $\mathbb{S}$ of size of $r$ from $\mathbb{P}$ according to the same distribution, and let $\epsilon>0$,
\begin{eqnarray}
Pr\bigg[\big|\sum_{P_l\in\mathbb{P}}\mathcal{M}(P_l, Q)-\frac{1}{r}\sum_{P_l\in\mathbb{S}}\frac{T}{t_{\mathbb{P}}(P_l)}\mathcal{M}(P_l, Q)\big|\leq \epsilon\sum_{P_l\in\mathbb{P}}\mathcal{M}(P_l, Q)\bigg]\geq 1-2 e^{-\frac{2r\epsilon^2}{T^2}}.\label{for-lang1}
\end{eqnarray}
\end{theorem}
In particular, (\ref{for-lang1}) is an application of Hoeffding's inequality because each $\frac{T}{t_{\mathbb{P}}(P_i)}\mathcal{M}(P_i, Q)$ is a random variable between $0$ and $T\sum_{P_l\in\mathbb{P}}\mathcal{M}(P_l, Q)$ (see Lemma 2.2 of~\cite{DBLP:conf/soda/VaradarajanX12} for more details). Moreover, (\ref{for-lang1}) shows that the sample $\mathbb{S}$ together with the weight $w_l=\frac{1}{r}\frac{T}{t_{\mathbb{P}}(P_l)}$ for each $P_l\in\mathbb{S}$ will form a core-set of $\mathbb{P}$ with respect to the fixed $Q$ (see (\ref{for-coreset})).
%
%
But  (\ref{for-coreset}) should hold for an infinity number of possible candidates for the geometric prototype, rather than one single $Q$, in the space. Hence, we need to determine an appropriate sample size (i.e., issue \textbf{(\Rmnum{2})}). 

Our basic idea is to discretize the space and generate a finite number of representations for them; then we can take a union bound for the final success probability through (\ref{for-lang1}). Note \cite{DBLP:conf/soda/VaradarajanX12} also used discretization to determine the sample size for projective clustering integer points; but our idea and analysis are quite different due to the different natures of the problems. Also, \cite{langberg2010universal,DBLP:conf/stoc/FeldmanL11} defined the ``dimension'' of the clustering problems so as to bounding their sample sizes. Here, we avoid using their approach due to two reasons: first, it will be very complicated to define and compute the dimension of geometric prototype problem; second, the framework in~\cite{DBLP:conf/stoc/FeldmanL11} would result in a more complicated sampling process and even may cause negative weights, however, we prefer to keep our sampling process simple as described in Theorem~\ref{the-lang} (especially when using any available algorithm or heuristic as a black box on the core-set). 
We elaborate on our analysis below.

Following Theorem~\ref{the-constant}, we assume that a randomly picked $P_{i_0}$ yields a $(2\alpha+2)$-approximation, and denote by $L$ the resulting cost $\sum_{P_l\in\mathbb{P}}\mathcal{M}(P_l, P_{i_0})$. The following lemma reveals that we just need to consider the $k$-point sets which are not too far from $P_{i_0}$.

\begin{lemma}
\label{lem-issue2-n1}
For any $k$-point set $Q$ with $\mathcal{M}(Q, P_{i_0})>\frac{4L}{n}$, the resulting cost $\sum_{P_l\in\mathbb{P}}\mathcal{M}(P_l, Q)$ is always higher than $\sum_{P_l\in\mathbb{P}}\mathcal{M}(P_l, P_{i_0})$.
\end{lemma}
\begin{proof}
Using Lemma~\ref{lem-tri}, we have
\begin{eqnarray}
\sum_{P_l\in\mathbb{P}}\mathcal{M}(P_l, Q)\geq\sum_{P_l\in\mathbb{P}}(\frac{1}{2}\mathcal{M}(Q, P_{i_0})-\mathcal{M}(P_l, P_{i_0}))>\frac{1}{2}4L-L=\sum_{P_l\in\mathbb{P}}\mathcal{M}(P_l, P_{i_0}).
\end{eqnarray}
So the proof is completed.
\end{proof}

Because we already have the initial solution $P_{i_0}$, we are only interested in the solutions having lower costs. Thus, we focus on the $k$-point set $Q$s with $\mathcal{M}(Q, P_{i_0})\leq\frac{4L}{n}$ based on Lemma~\ref{lem-issue2-n1}. 
Let $Q=\{q_1, q_2, \cdots, q_k\}\subset\mathbb{R}^d$ and $R=L/n$. 
W.l.o.g, we assume the induced permutation of $\mathcal{M}(Q, P_{i_0})$ in Definition~\ref{def-match} is $\pi(j)=j$ for $1\leq j\leq k$. The constraint $\mathcal{M}(Q, P_{i_0})\leq\frac{4L}{n}$ directly implies that $||q_j-p^{i_0}_j||\leq 2\sqrt{R}$ for each $1\leq j\leq k$. We use $B(x,\rho)$ to denote the ball centered at the point $x$ with the radius $\rho$. Then we draw $k$ balls $B(p^{i_0}_j,2\sqrt{R})$ for each $1\leq j\leq k$; inside each ball, we build a uniform grid $G_j$ with the grid side length $\epsilon\sqrt{\frac{R}{kd}}$. Let $\Gamma$ be the Cartesian product $G_1\times G_2\times\cdots\times G_k$. It is easy to know that $\Gamma$ contains $O\Big((\frac{4\sqrt{kd}}{\epsilon})^{kd}\Big)$ $k$-point sets in total. Therefore, we can apply (\ref{for-lang1}) of Theorem~\ref{the-lang} to obtain a union bound over all the $k$-point sets of $\Gamma$ (recall $T=O(1)$). 

%
%
%
\begin{lemma}
\label{lem-issue2-grid}
If the sample $\mathbb{S}$ in Theorem~\ref{the-lang} has the size of $O(\frac{kd}{\epsilon^2}\log\frac{kd}{\epsilon})$, and each $P_l\in\mathbb{S}$ has the weight $w_l=\frac{1}{r}\frac{T}{t_{\mathbb{P}}(P_l)}$, then with constant probability the inequality (\ref{for-coreset}) holds for each $Q\in\Gamma$.
\end{lemma}
\vspace{-0.1in}

Next we consider the $k$-point set $Q=\{q_1, q_2, \cdots, q_k\}\notin\Gamma$. Again, w.l.o.g, we assume the induced permutation of $\mathcal{M}(P_{i_0}, Q)$ is $\pi(j)=j$ for $1\leq j\leq k$. Also, due to our above assumption, we know that each $q_j$ is covered by the ball $B(p^{i_0}_j,2\sqrt{R})$. To help our analysis, we take its ``nearest neighbor'' from $\Gamma$, $\mathcal{N}(Q)=\{\mathcal{N}(q_1), \mathcal{N}(q_2), \cdots, \mathcal{N}(q_k)\}$ with each $\mathcal{N}(q_j)$ being the nearest grid point of $q_j$ in $G_j$. So we have 
\begin{eqnarray}
||q_j-\mathcal{N}(q_j)||\leq \epsilon\sqrt{\frac{R}{k}} \hspace{0.1in} for \hspace{0.1in} 1\leq j\leq k. \label{for-issue2-6}
\end{eqnarray}

It is easy to see $\Big|\sum_{P_l\in\mathbb{P}}\mathcal{M}(P_l, Q)-\frac{1}{r}\sum_{P_l\in\mathbb{S}}\frac{T}{t_{\mathbb{P}}(P_l)}\mathcal{M}(P_l, Q)\Big|\leq$
\begin{eqnarray}
\Big|\sum_{P_l\in\mathbb{P}}\mathcal{M}(P_l, Q)-\sum_{P_l\in\mathbb{P}}\mathcal{M}(P_l, \mathcal{N}(Q))\Big|+\Big|\sum_{P_l\in\mathbb{P}}\mathcal{M}(P_l, \mathcal{N}(Q))-\frac{1}{r}\sum_{P_l\in\mathbb{S}}\frac{T}{t_{\mathbb{P}}(P_l)}\mathcal{M}(P_l, \mathcal{N}(Q))\Big|\nonumber\\
+\Big|\frac{1}{r}\sum_{P_l\in\mathbb{S}}\frac{T}{t_{\mathbb{P}}(P_l)}\mathcal{M}(P_l, \mathcal{N}(Q))-\frac{1}{r}\sum_{P_l\in\mathbb{S}}\frac{T}{t_{\mathbb{P}}(P_l)}\mathcal{M}(P_l, Q)\Big|\hspace{0.2in}\label{for-largetri}
\end{eqnarray}
where the middle item is bounded by Lemma~\ref{lem-issue2-grid}. So the remaining issue is to prove that the other two items in (\ref{for-largetri}) 
%
%
are small as well. That is, Lemma~\ref{lem-issue2-grid} can be extended from $\mathcal{N}(Q)$ to $Q$.

Note $\Big|\sum_{P_l\in\mathbb{P}}\mathcal{M}(P_l, Q)-\sum_{P_l\in\mathbb{P}}\mathcal{M}(P_l, \mathcal{N}(Q))\Big|\leq \sum_{P_l\in\mathbb{P}}\big|\mathcal{M}(P_l, Q)-\mathcal{M}(P_l, \mathcal{N}(Q))\big|$, so we consider each $\big|\mathcal{M}(P_l, Q)-\mathcal{M}(P_l, \mathcal{N}(Q))\big|$ separately. Using Lemma~\ref{lem-tri2}, we have
\begin{eqnarray}
\big|\mathcal{M}(P_l, Q)-\mathcal{M}(P_l, \mathcal{N}(Q))\big|\leq (1+\frac{1}{\epsilon})\mathcal{M}(Q,\mathcal{N}(Q))+\epsilon\mathcal{M}(P_l, Q).\label{for-issue2-7}
\end{eqnarray}
In addition, we have $\mathcal{M}(Q, \mathcal{N}(Q))\leq k\Big(\epsilon\sqrt{\frac{R}{k}}\Big)^2=\epsilon^2 R$ by (\ref{for-issue2-6}). Therefore, we have 
\begin{eqnarray}
\Big|\sum_{P_l\in\mathbb{P}}\mathcal{M}(P_l, Q)-\sum_{P_l\in\mathbb{P}}\mathcal{M}(P_l, \mathcal{N}(Q))\Big|&\leq& \sum_{P_l\in\mathbb{P}}\big|\mathcal{M}(P_l, Q)-\mathcal{M}(P_l, \mathcal{N}(Q))\big|\nonumber\\
&\leq&(1+\frac{1}{\epsilon})n\mathcal{M}(Q, \mathcal{N}(Q))+\epsilon\sum_{P_l\in\mathbb{P}}\mathcal{M}(P_l, Q)\nonumber\\
&\leq& O(\epsilon)n R+\epsilon\sum_{P_l\in\mathbb{P}}\mathcal{M}(P_l,Q)\nonumber\\
&=&O(\epsilon)\sum_{P_l\in\mathbb{P}}\mathcal{M}(P_l, Q),\label{for-issue2-8}
\end{eqnarray}
where the last equality comes from $nR=L$ which is a constant approximation of the optimal objective value. (\ref{for-issue2-8}) also implies that
\begin{eqnarray}
\big(1-O(\epsilon)\big)\sum_{P_l\in\mathbb{P}}\mathcal{M}(P_l, Q)\leq\sum_{P_l\in\mathbb{P}}\mathcal{M}(P_l, \mathcal{N}(Q))\leq \big(1+O(\epsilon)\big)\sum_{P_l\in\mathbb{P}}\mathcal{M}(P_l, Q).\label{for-issue2-9}
\end{eqnarray}

Next, we consider the last item in (\ref{for-largetri}). 
It is a little more complicated because the coefficient $\frac{T}{t_{\mathbb{P}}(P_l)}$ could be large. We need the following lemma first.
\vspace{-0.1in}
\begin{lemma}
\label{lem-lowsens}
For each $P_l\in\mathbb{P}$, $t_{\mathbb{P}}(P_l)>\frac{1}{4n}$.
\end{lemma}
\begin{proof}
Fix one $P_l\in\mathbb{P}$. We select $P_{l'}$ that has the largest matching cost to $P_l$, i.e., $\mathcal{M}(P_l, P_{l'})=\max_{P_i\in\mathbb{P}}\mathcal{M}(P_l, P_i)$, and set $Q=P_{l'}$. Using Lemma~\ref{lem-tri}, we have $\mathcal{M}(P_i, Q)\leq 2\mathcal{M}(P_i, P_l)+2\mathcal{M}(P_l, Q)\leq 4\mathcal{M}(P_l, Q)$ for any $1\leq i\leq n$. Therefore, based on the fact that $t_{\mathbb{P}}(P_l)$ is the upper bound of $\sigma_{\mathbb{P}}(P_l)$ in (\ref{for-sen1}), we know that it should be at least $\frac{\mathcal{M}(P_l, Q)}{(1+4(n-1))\mathcal{M}(P_l, Q)}>\frac{1}{4n}$.
\end{proof}
\vspace{-0.1in}

Using Lemma~\ref{lem-lowsens} and the same idea for (\ref{for-issue2-8}), we have
\begin{eqnarray}
&&\Big|\frac{1}{r}\sum_{P_l\in\mathbb{S}}\frac{T}{t_{\mathbb{P}}(P_l)}\mathcal{M}(P_l, \mathcal{N}(Q))-\frac{1}{r}\sum_{P_l\in\mathbb{S}}\frac{T}{t_{\mathbb{P}}(P_l)}\mathcal{M}(P_l, Q)\Big|\nonumber\\
&\leq&\frac{1}{r} \sum_{P_l\in\mathbb{S}}\frac{T}{t_{\mathbb{P}}(P_l)}\Big|\mathcal{M}(P_l, \mathcal{N}(Q))-\mathcal{M}(P_l, Q)\Big|\nonumber\\
&\leq &\frac{1}{r}\sum_{P_l\in\mathbb{S}}\frac{T}{t_{\mathbb{P}}(P_l)}\Big((1+\frac{1}{\epsilon})\mathcal{M}(\mathcal{N}(Q),Q)+\epsilon \mathcal{M}(P_l, \mathcal{N}(Q))\Big)\nonumber\\
&\leq&\max_{P_l\in\mathbb{S}}\{\frac{T}{t_{\mathbb{P}}(P_l)}\}\cdot(1+\frac{1}{\epsilon})\mathcal{M}(\mathcal{N}(Q),Q)+\epsilon \frac{1}{r}\sum_{P_l\in\mathbb{S}}\frac{T}{t_{\mathbb{P}}(P_l)}  \mathcal{M}(P_l, \mathcal{N}(Q))\nonumber\\
&\leq&O(\epsilon)nR+\epsilon \frac{1}{r}\sum_{P_l\in\mathbb{S}}\frac{T}{t_{\mathbb{P}}(P_l)}  \mathcal{M}(P_l, \mathcal{N}(Q)),\label{for-issue2-10}
\end{eqnarray}
where the last inequality comes from Lemma~\ref{lem-lowsens} and $T=O(1)$. 
In addition, Lemma~\ref{lem-issue2-grid} guarantees that $\epsilon\frac{1}{r}\sum_{P_l\in\mathbb{S}}\frac{T}{t_{\mathbb{P}}(P_l)}\mathcal{M}(P_l, \mathcal{N}(Q))=O(\epsilon)\sum_{P_l\in \mathbb{P}}\mathcal{M}(P_l, \mathcal{N}(Q))$. 
Applying the triangle inequality (\ref{for-largetri}) with the bounds (\ref{for-issue2-8}), (\ref{for-issue2-9}) and (\ref{for-issue2-10}), we have
\begin{eqnarray}
\Big|\sum_{P_l\in\mathbb{P}}\mathcal{M}(P_l,Q)-\frac{1}{r}\sum_{P_l\in\mathbb{S}}\frac{T}{t_{\mathbb{P}}(P_l)}\mathcal{M}(P_l,Q)\Big|
&\leq& O(\epsilon)\sum_{P_l\in\mathbb{P}}\mathcal{M}(P_l,Q)+O(\epsilon)\sum_{P_l\in\mathbb{P}}\mathcal{M}(P_l, \mathcal{N}(Q))\nonumber\\
&=&O(\epsilon)\sum_{P_l\in\mathbb{P}}\mathcal{M}(P_l, Q).\label{for-issue2-11}
\end{eqnarray}
Consequently, we have the final theorem for core-set.

\begin{theorem}
\label{the-coreset}
Let $P_{i_0}$ be the $k$-point set randomly selected by Theorem~\ref{the-constant}, and $\mathbb{S}$ be the sample from $\mathbb{P}$ according to the distribution $\frac{t_{\mathbb{P}}(P_i)}{T}$. If the sample $\mathbb{S}$ has the size of $r=O(\frac{kd}{\epsilon^2}\log\frac{kd}{\epsilon})$ and  each $P_l\in\mathbb{S}$ has the weight $w_l=\frac{1}{r}\frac{T}{t_{\mathbb{P}}(P_l)}$, then with constant probability the inequality (\ref{for-coreset}) holds for any $k$-point set $Q\subset\mathbb{R}^d$ with $\mathcal{M}(Q, P_{i_0})\leq\frac{4L}{n}$.
\end{theorem}

Recall that Theorem~\ref{the-jl} tells us that the dimension can be reduced by Johnson-Lindenstrauss (JL)-transform. Thus, we directly have the following corollary. 

\begin{corollary}
\label{cor-coreset}
Given a high dimensional instance $\mathbb{P}$, we can obtain a sample $\mathbb{S}$ having the size of $\tilde{O}(\frac{k}{\epsilon^4})$, where with constant probability the inequality (\ref{for-coreset}) holds for any $k$-point set $Q\subset\mathbb{R}^d$ with $\mathcal{M}(Q, P_{i_0})\leq\frac{4L}{n}$. $\tilde{O}(\cdot)$ ignores logarithmic factors.
\end{corollary}

\subsection{Some Extensions}
\label{sec-coreext}


Here, we briefly introduce some extensions of our core-set construction on other metrics. Due to space limit, more details are shown in Appendix.

\noindent\textbf{(1).} Our core-set construction can be extended to $l_1$ norm, i.e., the squared distances are replaced by absolute distances in the matching cost~(\ref{for-match}). Actually, the analysis for $l_1$ norm is even easier than that for $l_2$ norm, since we can directly use triangle inequality rather than Lemma~\ref{lem-tri} or Lemma~\ref{lem-tri2} when solving the aforementioned two issues, bounding the sensitivities and discretizing the space of candidates for geometric prototype. 

A remaining issue for future work is that the dimension reduction result of Theorem~\ref{the-jl} is not applicable to $l_1$ norm, due to the fact that it is much harder to compute geometric median (Fermat-Weber point) than mean point~\cite{cohen2016geometric}. Fortunately, the high dimensional application, ensemble clustering, mentioned in Section~\ref{sec-intro} only uses $l_2$ norm, because the symmetric difference between two clusters corresponds to their squared distance in the space.

\noindent\textbf{(2).} We can also consider the case with weighted point sets for both $l_1$ and $l_2$ norm, i.e., each point of $P_i$ has a non-negative weight. To make the problem meaningful in practice, we require that each $P_i$ and the desired geometric prototype have the same total weight $W>0$; we can further assume $W$ and all the weights are integers by scaling and rounding in practice. Thus, the computation on the matching between two point sets becomes the problem of earth mover's distance (EMD)~\cite{rubner2000earth}. Fortunately, triangle inequality still holds for EMD because we assume they have equal total weight; as a consequence, we can bound the sensitivities for issue \textbf{(\Rmnum{1})}. For issue \textbf{(\Rmnum{2})}, we still discretize the space and build the set of $k$-point sets $\Gamma$ with the same cardinality of the unweighted case; the only difference is that we need to consider the total $O(W^k)$ possible distributions of the total weight $W$ over the $k$ points of each $k$-point set, which increases the size of the core-set with an extra $O(\frac{k\log W}{\epsilon^2})$.

\subsection{The Time Complexity}
\label{sec-time}

Suppose the complexity of computing $\mathcal{M}(A, B)$ is $h(k, d)$, then the running time for computing the core-set is simply $O(h(k, d)\cdot n)$ because we just need to compute each $\mathcal{M}(P_i, P_{i_0})$ so as to obtain the sensitivities for sampling (see Lemma~\ref{lem-issue12}). For simplicity, we can just use Hungarian algorithm~\cite{cormenintroduction} and $h(k, d)=O(k^2d+k^3)$, where the term $k^2d$ is for building the bipartite graph. In fact, this can be further improved by our following two observations. First, we just need to know the matching costs, rather than the matchings, for computing the upper bounds of the sensitivities in Lemma~\ref{lem-issue12}. Second, it is not necessary to always have the optimal matching costs. For example, if we compute a value $\mathcal{M}'(P_i, P_{i_0})$ for each $\mathcal{M}(P_i, P_{i_0})$ instead, such that $\mathcal{M}(P_i, P_{i_0})\leq \mathcal{M}'(P_i, P_{i_0})\leq c\mathcal{M}(P_i, P_{i_0})$ with some constant $c\geq1$, the resulting $T$ and each $t_{\mathbb{P}}(P_i)$ will increase by some appropriate constant factors correspondingly; in other words, the sample size in Theorem~\ref{the-coreset} will increase by only a constant factor. Some algorithms~\cite{cabello2008matching,indyk2007near,IT03} are designed for approximately estimating the matching cost, and their running times can  be nearly linear if the dimension $d$ is constant; in practical fields, several heuristic algorithms~\cite{pele2009fast} are also proposed for this purpose. 

For high dimensional case, we can apply JL-transform in advance, to reduce the dimensionality to be $O(\log(nk)/\epsilon^2)$ (Theorem~\ref{the-jl} and Corollary~\ref{cor-coreset}). A naive implementation of JL-transform by matrix multiplication has the complexity $O\big(\frac{1}{\epsilon^2}nkd\log (nk)\big)$~\cite{dasgupta2003elementary}, and several even faster and practical algorithms have been studied before~\cite{achlioptas2003database,liberty2009mailman,ailon2009fast}.

\section{Experiments}
\label{sec-exp}

To show the advantage of using core-set for the problem of geometric prototype, we study the two important applications introduced in Section~\ref{sec-intro}, Wasserstein barycenter and ensemble clustering. For each application, we run the most recent state of the art algorithm on the original dataset and core-sets with different size levels. In general, our experiments suggest that running the algorithm on a small core-set can achieve very close performance and greatly reduce the running time. 
All of the experimental results are obtained on a Windows workstation with 2.4GHz Intel Xeon E5-2630 v3 CPU and 32GB DDR4 2133MHz Memory; the algorithms are implemented in Matlab R2016b.
 
 
\textbf{Wasserstein barycenter.} MNIST~\cite{lecun1998gradient} is a popular benchmark dataset of handwritten digits from $0$ to $9$. For each digit, we generate a set of 3000 $28\times 28$ grayscale images including $10\%$ noise (i.e., $300$ images randomly selected from the other $9$ digits). First, we represent the $28\times 28$ pixels by $60$ weighted $2D$ points via $k$-means clustering~\cite{lloyd1982least}: group the pixels into $60$ clusters and each cluster is represented by its cluster center; each center has the weight equal to the total pixel values of the cluster. Therefore the problem of Wasserstein barycenter becomes an instance of geometric prototype with $n=3000$, $k=60$, and $d=2$. 
 
\textbf{Ensemble clustering.} To construct an instance of ensemble clustering, we generate a synthetic dataset of $2000$ points randomly sampled from $k=50$ Gaussian distributions in $\mathbb{R}^{100}$; we apply $k$-means clustering $1000$ times, where each time has a different initialization for the $k$ mean points, to generate $1000$ different clustering solutions. 
According to the model introduced by~\cite{DBLP:conf/mobihoc/DingSX16}, each instance is a geometric prototype problem with $1000$ different $50$-point sets in $\mathbb{R}^{2000}$. We apply JL-transform to reduce the dimensionality from $2000$ to $100$, before constructing the core-set and running the algorithm; we just use the simplest random matrix multiplication to implement JL-transform~\cite{dasgupta2003elementary}, where actually this step only takes about $5\%$ of the whole running time of the experiments.

\begin{figure}[]
\begin{minipage}[t]{0.5\linewidth}
  \centering
\includegraphics[scale=0.22]{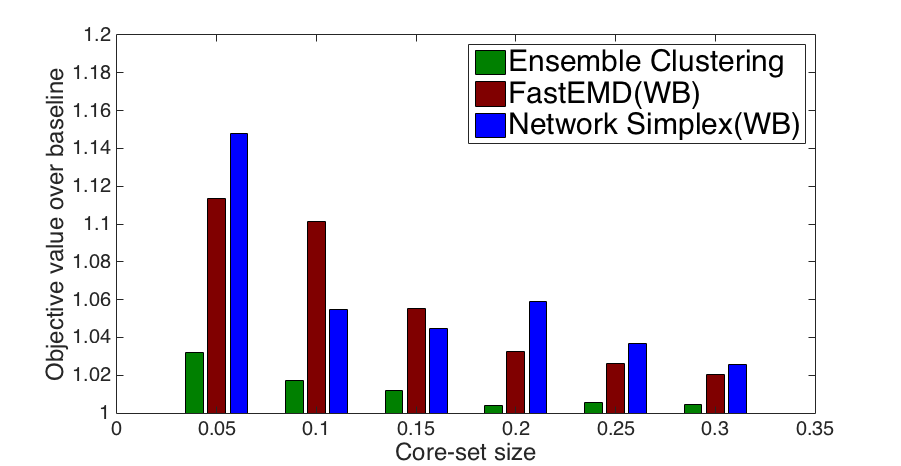}
  \vspace{-0.15in}
     \caption{Normalized objective value.}
  \label{fig-obj}
\end{minipage}
\hspace{-0.15in}
\begin{minipage}[t]{0.5\linewidth}
\centering
\includegraphics[scale=0.22]{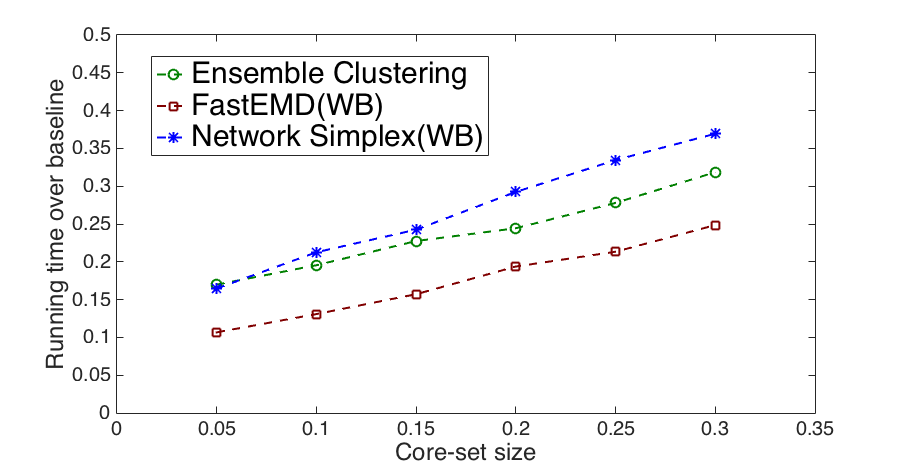}
  \vspace{-0.15in}
\caption{Normalized running time.}
  \label{fig-time}
  \end{minipage}
  \begin{minipage}[t]{0.5\linewidth}
  \centering
\includegraphics[scale=0.22]{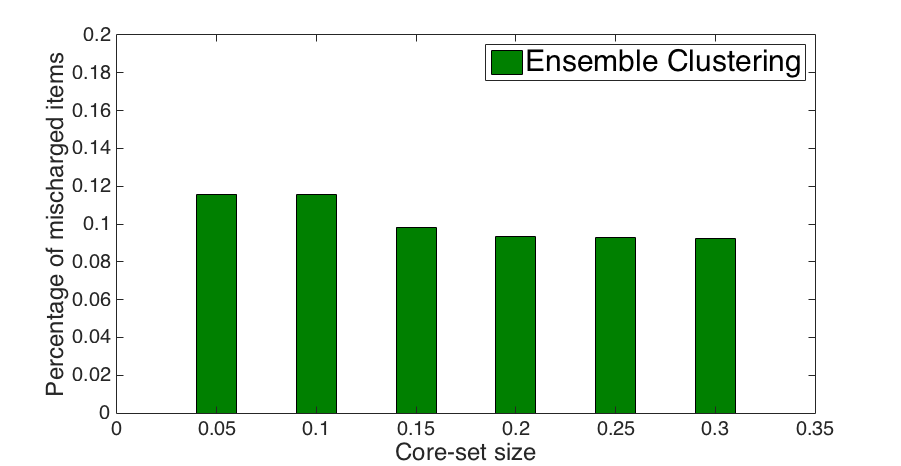}
  \vspace{-0.15in}
     \caption{Percentage of misclustered items.}
  \label{fig-dis1}
\end{minipage}
\hspace{-0.15in}
\begin{minipage}[t]{0.5\linewidth}
\centering
\includegraphics[scale=0.22]{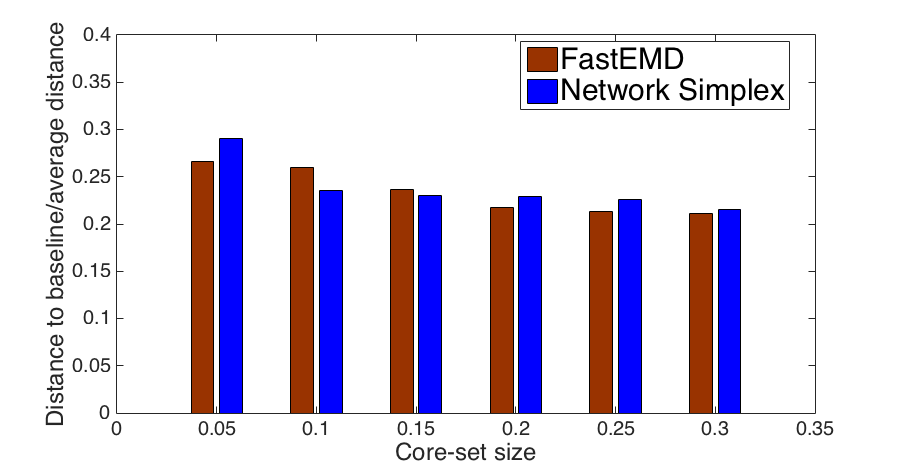}
  \vspace{-0.15in}
\caption{Matching cost to ground truth.}
  \label{fig-dis2}
  \end{minipage}
  \vspace{-0.15in}
\end{figure}

For both applications, we construct the core-sets using the method in Section~\ref{sec-coreset}; we vary the core-set size from $5\%$ to $30\%$ of the input size. To construct the core-set, we need to compute the matching cost $\mathcal{M}(P_i, P_{i_0})$ as discussed in Section~\ref{sec-time}: for the high dimensional application (i.e., ensemble clustering), we just use Hungarian algorithm~\cite{cormenintroduction}; for the low dimensional application (i.e., Wasserstein barycenter), we use two existing popular algorithms for computing EMD, {\em Network simplex algorithm}~\cite{ahuja1993network} and the heuristic but faster EMD algorithm~\cite{pele2009fast}. 
As the black boxes, we use the algorithms in~\cite{ye2017fast} and~\cite{DBLP:conf/mobihoc/DingSX16} for Wasserstein barycenter and Ensemble clustering, respectively. For each application, we run the same algorithm on the original input dataset and corresponding core-sets, and consider three criteria: running time, objective value (in Definition~\ref{def-gp}), and difference to ground truth.  For ensemble clustering, we compute the percentage of misclustered items of the obtained prototype as the difference to ground truth. For Wasserstein barycenter, since it is difficult to determine a unique ground truth for each handwritten digit, we directly use the prototype obtained from the original input dataset as the ground truth; then we compute its matching cost to the prototype obtained from core-set, denoted by $x$, as well as the average matching cost over the input images to the ground truth, denoted by $Ave$; finally, we obtain the ratio $x/Ave$. In general, the lower the ratio $x/Ave$, the closer the obtained prototype to the ground truth (comparing with the input images).

\textbf{Results.} For each application, we run $50$ trials and report the average results.
 Figure~\ref{fig-obj} shows the obtained normalized objective values over the base line (i.e., the objective value obtained on the original input dataset), which are all lower than $1.2$; that means our core-sets are good approximations for the original data. More importantly, the running times are significantly reduced in Figure~\ref{fig-time}, e.g., for the core-set having $5\%$ of the input data size, the algorithm (containing the core-sets construction) only runs within $10\%$-$17\%$ of the original time. In addition, our obtained prototypes are very close to the corresponding ground truths, even for the core-set at the level $5\%$. Figure~\ref{fig-dis1} provides the percentages of misclustered items for ensemble clustering, which are around $8\%$-$12\%$. Figure~\ref{fig-dis2} shows the values of $x/Ave$, which are around $0.25$. For Wasserstein barycenter, we can see Network simplex algorithm and fast EMD algorithm achieve very similar qualities, but fast EMD only takes about $60\%$ of the running time of Network simplex algorithm.


%

\newpage

\bibliography{nips_2017}

\begin{thebibliography}{10}

\bibitem{achlioptas2003database}
Dimitris Achlioptas.
\newblock Database-friendly random projections: Johnson-lindenstrauss with
  binary coins.
\newblock {\em Journal of computer and System Sciences}, 66(4):671--687, 2003.

\bibitem{DBLP:conf/compgeom/AgarwalFPVX17}
Pankaj~K. Agarwal, Kyle Fox, Debmalya Panigrahi, Kasturi~R. Varadarajan, and
  Allen Xiao.
\newblock Faster algorithms for the geometric transportation problem.
\newblock In {\em 33rd International Symposium on Computational Geometry, SoCG
  2017, July 4-7, 2017, Brisbane, Australia}, pages 7:1--7:16, 2017.

\bibitem{agarwal2005geometric}
Pankaj~K. Agarwal, Sariel Har-Peled, and Kasturi~R. Varadarajan.
\newblock Geometric approximation via coresets.
\newblock {\em Combinatorial and Computational Geometry}, 52:1--30, 2005.

\bibitem{DBLP:conf/compgeom/AgarwalV04}
Pankaj~K. Agarwal and Kasturi~R. Varadarajan.
\newblock A near-linear constant-factor approximation for euclidean bipartite
  matching?
\newblock In {\em Proceedings of the 20th {ACM} Symposium on Computational
  Geometry, Brooklyn, New York, USA, June 8-11, 2004}, pages 247--252, 2004.

\bibitem{ahuja1993network}
Ravindra~K Ahuja, Thomas~L Magnanti, and James~B Orlin.
\newblock {\em Network flows: theory, algorithms, and applications}.
\newblock Prentice Hall, 1993.

\bibitem{ailon2009fast}
Nir Ailon and Bernard Chazelle.
\newblock The fast johnson--lindenstrauss transform and approximate nearest
  neighbors.
\newblock {\em SIAM Journal on computing}, 39(1):302--322, 2009.

\bibitem{DBLP:conf/stoc/AndoniNOY14}
Alexandr Andoni, Aleksandar Nikolov, Krzysztof Onak, and Grigory Yaroslavtsev.
\newblock Parallel algorithms for geometric graph problems.
\newblock In {\em Symposium on Theory of Computing, {STOC} 2014, New York, NY,
  USA, May 31 - June 03, 2014}, pages 574--583, 2014.

\bibitem{arkin2015bichromatic}
Esther~M Arkin, Jos{\'e}~Miguel D{\'\i}az-B{\'a}{\~n}ez, Ferran Hurtado, Piyush
  Kumar, Joseph~S.B. Mitchell, Bel{\'e}n Palop, Pablo P{\'e}rez-Lantero, Maria
  Saumell, and Rodrigo~I Silveira.
\newblock Bichromatic 2-center of pairs of points.
\newblock {\em Computational Geometry}, 48(2):94--107, 2015.

\bibitem{baum2015wasserstein}
Marcus Baum, Peter Willett, and Uwe~D. Hanebeck.
\newblock On wasserstein barycenters and {MMOSPA} estimation.
\newblock {\em {IEEE} Signal Process. Lett.}, 22(10):1511--1515, 2015.

\bibitem{DBLP:journals/siamsc/BenamouCCNP15}
Jean-David Benamou, Guillaume Carlier, Marco Cuturi, Luca Nenna, and Gabriel
  Peyr{\'e}.
\newblock Iterative bregman projections for regularized transportation
  problems.
\newblock {\em SIAM Journal on Scientific Computing}, 37(2):A1111--A1138, 2015.

\bibitem{bonizzoni2008approximation}
Paola Bonizzoni, Gianluca Della~Vedova, Riccardo Dondi, and Tao Jiang.
\newblock On the approximation of correlation clustering and consensus
  clustering.
\newblock {\em Journal of Computer and System Sciences}, 74(5):671--696, 2008.

\bibitem{boyd2011distributed}
Stephen Boyd, Neal Parikh, Eric Chu, Borja Peleato, and Jonathan Eckstein.
\newblock Distributed optimization and statistical learning via the alternating
  direction method of multipliers.
\newblock {\em Foundations and Trends{\textregistered} in Machine learning},
  3(1):1--122, 2011.

\bibitem{cabello2008matching}
Sergio Cabello, Panos Giannopoulos, Christian Knauer, and G{\"u}nter Rote.
\newblock Matching point sets with respect to the earth mover's distance.
\newblock {\em Computational Geometry}, 39(2):118--133, 2008.

\bibitem{chen2009coresets}
Ke~Chen.
\newblock On coresets for k-median and k-means clustering in metric and
  euclidean spaces and their applications.
\newblock {\em SIAM Journal on Computing}, 39(3):923--947, 2009.

\bibitem{cohen2016geometric}
Michael~B Cohen, Yin~Tat Lee, Gary Miller, Jakub Pachocki, and Aaron Sidford.
\newblock Geometric median in nearly linear time.
\newblock In {\em Proceedings of the 48th Annual ACM SIGACT Symposium on Theory
  of Computing}, pages 9--21. ACM, 2016.

\bibitem{cormenintroduction}
Thomas~H. Cormen, Charles~E. Leiserson, Ronald~L. Rivest, and Clifford Stein.
\newblock {\em Introduction to Algorithms, Third Edition}.
\newblock The MIT Press, 3rd edition, 2009.

\bibitem{cuturi2014fast}
Marco Cuturi and Arnaud Doucet.
\newblock Fast computation of wasserstein barycenters.
\newblock In {\em International Conference on Machine Learning}, pages
  685--693, 2014.

\bibitem{D08}
Sanjoy Dasgupta.
\newblock The hardness of k-means clustering.
\newblock {\em Technical Report}, 2008.

\bibitem{dasgupta2003elementary}
Sanjoy Dasgupta and Anupam Gupta.
\newblock An elementary proof of a theorem of johnson and lindenstrauss.
\newblock {\em Random Structures \& Algorithms}, 22(1):60--65, 2003.

\bibitem{ding2013k}
Hu~Ding, Ronald Berezney, and Jinhui Xu.
\newblock k-prototype learning for 3d rigid structures.
\newblock In {\em Advances in Neural Information Processing Systems}, pages
  2589--2597, 2013.

\bibitem{DBLP:conf/mobihoc/DingSX16}
Hu~Ding, Lu~Su, and Jinhui Xu.
\newblock Towards distributed ensemble clustering for networked sensing
  systems: a novel geometric approach.
\newblock In {\em Proceedings of the 17th {ACM} International Symposium on
  Mobile Ad Hoc Networking and Computing, MobiHoc 2016, Paderborn, Germany,
  July 4-8, 2016}, pages 1--10, 2016.

\bibitem{Din11}
Hu~Ding and Jinhui Xu.
\newblock Solving the chromatic cone clustering problem via minimum spanning
  sphere.
\newblock In {\em Proceedings of the International Colloquium on Automata,
  Languages, and Programming (ICALP)}, pages 773--784, 2011.

\bibitem{ding2014finding}
Hu~Ding and Jinhui Xu.
\newblock Finding median point-set using earth mover's distance.
\newblock In {\em Twenty-Eighth AAAI Conference on Artificial Intelligence},
  2014.

\bibitem{DX15}
Hu~Ding and Jinhui Xu.
\newblock A unified framework for clustering constrained data without locality
  property.
\newblock In {\em Proceedings of the Twenty-Sixth Annual ACM-SIAM Symposium on
  Discrete Algorithms}, pages 1471--1490, 2015.

\bibitem{DBLP:conf/stoc/FeldmanL11}
Dan Feldman and Michael Langberg.
\newblock A unified framework for approximating and clustering data.
\newblock In {\em Proceedings of the 43rd {ACM} Symposium on Theory of
  Computing, {STOC} 2011, San Jose, CA, USA, 6-8 June 2011}, pages 569--578,
  2011.

\bibitem{GLF}
Jing Gao, Feng Liang, Wei Fan, Yizhou Sun, and Jiawei Han.
\newblock Graph-based consensus maximization among multiple supervised and
  unsupervised models.
\newblock In {\em Advances in Neural Information Processing Systems}, pages
  585--593, 2009.

\bibitem{GA13}
Joydeep Ghosh and Ayan Acharya.
\newblock Cluster ensembles: Theory and applications.
\newblock In {\em Data Clustering: Algorithms and Applications}, pages
  551--570. 2013.

\bibitem{goelflow}
Ashish Goel.
\newblock Introduction to optimization.
\newblock {\em MS \& E 111, ENGR 62 Autumn 2008-2009, Course Lecture}.

\bibitem{gramfort2015fast}
Alexandre Gramfort, Gabriel Peyr{\'e}, and Marco Cuturi.
\newblock Fast optimal transport averaging of neuroimaging data.
\newblock In {\em International Conference on Information Processing in Medical
  Imaging}, pages 261--272. Springer, 2015.

\bibitem{indyk2007near}
Piotr Indyk.
\newblock A near linear time constant factor approximation for euclidean
  bichromatic matching (cost).
\newblock In {\em Proceedings of the eighteenth annual ACM-SIAM symposium on
  Discrete algorithms}, pages 39--42. Society for Industrial and Applied
  Mathematics, 2007.

\bibitem{IT03}
Piotr Indyk and Nitin Thaper.
\newblock Fast color image retrieval via embeddings.
\newblock In {\em Workshop on Statistical and Computational Theories of Vision
  (at ICCV)}, 2003.

\bibitem{langberg2010universal}
Michael Langberg and Leonard~J Schulman.
\newblock Universal $\varepsilon$-approximators for integrals.
\newblock In {\em Proceedings of the twenty-first annual ACM-SIAM symposium on
  Discrete Algorithms}, pages 598--607. SIAM, 2010.

\bibitem{lecun1998gradient}
Yann LeCun, L{\'e}on Bottou, Yoshua Bengio, and Patrick Haffner.
\newblock Gradient-based learning applied to document recognition.
\newblock {\em Proceedings of the IEEE}, 86(11):2278--2324, 1998.

\bibitem{liberty2009mailman}
Edo Liberty and Steven~W Zucker.
\newblock The mailman algorithm: A note on matrix--vector multiplication.
\newblock {\em Information Processing Letters}, 109(3):179--182, 2009.

\bibitem{GHL}
Jialu Liu, Chi Wang, Jing Gao, and Jiawei Han.
\newblock Multi-view clustering via joint nonnegative matrix factorization.
\newblock In {\em Proc. of SDM}, volume~13, pages 252--260, 2013.

\bibitem{lloyd1982least}
Stuart Lloyd.
\newblock Least squares quantization in pcm.
\newblock {\em IEEE transactions on information theory}, 28(2):129--137, 1982.

\bibitem{pele2009fast}
Ofir Pele and Michael Werman.
\newblock Fast and robust earth mover's distances.
\newblock In {\em Computer vision, 2009 IEEE 12th international conference on},
  pages 460--467. IEEE, 2009.

\bibitem{DBLP:journals/corr/Phillips16}
Jeff~M. Phillips.
\newblock Coresets and sketches.
\newblock {\em Computing Research Repository}, 2016.

\bibitem{rubner2000earth}
Yossi Rubner, Carlo Tomasi, and Leonidas~J Guibas.
\newblock The earth mover's distance as a metric for image retrieval.
\newblock {\em International journal of computer vision}, 40(2):99--121, 2000.

\bibitem{DBLP:conf/soda/SharathkumarA12}
R.~Sharathkumar and Pankaj~K. Agarwal.
\newblock Algorithms for the transportation problem in geometric settings.
\newblock In {\em Proceedings of the Twenty-Third Annual {ACM-SIAM} Symposium
  on Discrete Algorithms, {SODA} 2012, Kyoto, Japan, January 17-19, 2012},
  pages 306--317, 2012.

\bibitem{DBLP:conf/stoc/SharathkumarA12}
R.~Sharathkumar and Pankaj~K. Agarwal.
\newblock A near-linear time {\(\epsilon\)}-approximation algorithm for
  geometric bipartite matching.
\newblock In {\em Proceedings of the 44th Symposium on Theory of Computing
  Conference, {STOC} 2012, New York, NY, USA, May 19 - 22, 2012}, pages
  385--394, 2012.

\bibitem{SMP}
Vikas Singh, Lopamudra Mukherjee, Jiming Peng, and Jinhui Xu.
\newblock Ensemble clustering using semidefinite programming with applications.
\newblock {\em Machine learning}, 79(1-2):177--200, 2010.

\bibitem{staib2017parallel}
Matthew Staib, Sebastian Claici, Justin Solomon, and Stefanie Jegelka.
\newblock Parallel streaming wasserstein barycenters.
\newblock {\em arXiv preprint arXiv:1705.07443}, 2017.

\bibitem{SG02}
Alexander Strehl and Joydeep Ghosh.
\newblock Cluster ensembles-a knowledge reuse framework for combining
  partitionings.
\newblock In {\em AAAI/IAAI}, pages 93--99, 2002.

\bibitem{DBLP:conf/soda/VaradarajanX12}
Kasturi~R. Varadarajan and Xin Xiao.
\newblock A near-linear algorithm for projective clustering integer points.
\newblock In {\em Proceedings of the Twenty-Third Annual {ACM-SIAM} Symposium
  on Discrete Algorithms, {SODA} 2012, Kyoto, Japan, January 17-19, 2012},
  pages 1329--1342, 2012.

\bibitem{ye2017fast}
Jianbo Ye, Panruo Wu, James~Z Wang, and Jia Li.
\newblock Fast discrete distribution clustering using wasserstein barycenter
  with sparse support.
\newblock {\em IEEE Transactions on Signal Processing}, 65(9):2317--2332, 2017.

\end{thebibliography}

\newpage

\section{Appendix}

\subsection{The Hardness Proof for Geometric Prototype}
\label{sec-hard}

\begin{theorem}
\label{mfptas}
Finding the geometric prototype of a given instance is NP-hard and has no FPTAS even if $k=2$ in high dimensional space, unless P=NP.
\end{theorem}
\begin{proof}

It is sufficient to only consider the simplest case $k=2$ in our proof. As mentioned in the remark below Definition~\ref{def-gp}, the problem of geometric prototype is equivalent to finding a chromatic partition on the $2n$ points $\cup^n_{i=1}\{p^i_1, p^i_2\}$ to form two clusters, such that the sum of their variances is minimized. 
%
%
Based on this observation, we make use of the construction by Dasgupta for the
NP-hardness proof of $2$-means clustering problem in high dimension~\cite{D08}. His proof reduces from the {\em NAE3SAT} problem. The setting of NAE3SAT is similar to 3-SAT, where the only difference is that it requires at least one literal to be true and at least one literal to be false in each clause. Dasgupta considers a special case of NAE3SAT, denoted by NAE3SAT$^*$ (we refer the reader to \cite{D08} for more details on NAE3SAT$^*$). 
For better understanding our
ideas, below we sketch his construction.
For any instance $\phi$ of NAE3SAT$^*$ with  variables $\{x_{1}, \cdots, x_{n}\}$ and $m$ clauses,  construct a $2n \times 2n$ matrix $D$ as follows. For each entry $D_{\alpha, \beta}$, the indices  correspond to  $\{x_{1}, \cdots, x_{n}\}$ when they are in the range of $[1,n]$, and to $\{\overline{x}_{1}, \cdots, \overline{x}_{n}\}$ when they are in the range of $[n+1,2n]$.

$$D_{\alpha, \beta} = \left\{ \begin{array}{ll}
0 & \textrm{if $\alpha=\beta$}\\
1+\Delta & \textrm{if $\alpha=\overline{\beta}$}\\
1+\delta & \textrm{if $\alpha\sim\beta$}\\
1 & \textrm{otherwise,}
\end{array} \right.$$
where $\Delta$, $\delta$ are two constants satisfying inequalities $0<\delta<\Delta<1$ and $4\delta m<\Delta\leq 1-2\delta n$, and $\alpha\sim\beta$ means that either $\alpha$ and $\beta$ or $\overline{\alpha}$ and $\overline{\beta}$ appear together in a clause. 
Also, $D$ can be embedded into $\mathbb{R}^{2n}$, i.e., there exist $2n$ points in $\mathbb{R}^{2n}$ with $D$ as their pairwise squared distance matrix~\cite{D08}. With a slight abuse of notation, we also use $x_i$ and $\overline{x}_i$ to denote their corresponding embedding points. 
Let $C_{1}$ and $C_{2}$ be the two clusters of any $2$-mean clustering of the $2n$ embedding points. Then \cite{D08} provides the following two important claims:
\begin{enumerate}
\item If $C_1$ or $C_2$ contains both  $x_i$ and its negation $\overline{x}_i$ for some $i$, the clustering cost is larger than $n-1+\frac{2\delta m}{n}$.

\item $\phi$ is satisfiable if and only if the $2n$ embedding points admit a $2$-mean clustering having the cost no more than $n-1+\frac{2\delta m}{n}$.
\end{enumerate}

The second claim directly implies the NP-hardness of $2$-mean clustering. Moreover, a byproduct is that 
the above two claims jointly ensure that $\phi$ is satisfiable if and only if each of $C_1$ and $C_2$ contains exactly $n$ points, i.e., a chromatic partition on $\cup^n_{i=1}\{x_i, \overline{x}_i\}$, and the clustering cost is no more than $n-1+\frac{2\delta m}{n}$. Therefore, our geometric prototype problem is also NP-hard.

%

%

Next, we show that the problem of geometric prototype has no FPTAS in high dimensional space unless P=NP. To see this, we still use the same construction. From the definition of NAE3SAT$^*$,  we know that $\phi$ is unsatisfiable if and only if for any chromatic partition of $\cup^n_{i=1}\{x_{i}, \overline{x_{i}}\}$, there exists one clause in $\phi$ such that the three points corresponding to the three literals in this clause are clustered into the same cluster. Recall the fact that given a set of points $A=\{a_1, a_2, \cdots, a_h\}$ in Euclidean space, 
\begin{eqnarray}
\sum^h_{i=1}||a_i-\mu_A||^2=\frac{1}{2h}\sum^h_{i=1}\sum^h_{j=1}||a_i-a_j||^2\label{for-fact1}
\end{eqnarray}
where $\mu_A=\frac{1}{h}\sum^h_{i=1}a_i$~\cite{D08}. 
Hence, based on the construction of the matrix $D$ and (\ref{for-fact1}), we know the total clustering cost for any chromatic partition is at least
\begin{eqnarray}
&&\frac{2}{2n}\Bigg(2{n\choose 2}+2\delta\sum_{clauses}(\textnormal{1 if clause is split between two clusters; 3 otherwise})\Bigg)\nonumber\\
&\geq& \frac{2}{2n}\Bigg(2{n\choose 2}+2\delta\big((m-1)+3\big)\Bigg)=n-1+\frac{2}{n}(m+2)\delta
\end{eqnarray}
if $\phi$ is unsatisfiable. 
Thus, the ratio between the minimum chromatic partition cost of an unsatisfiable instance and the upper bound cost of a satisfiable instance is  
\begin{eqnarray}
\eta=\frac{n-1+\frac{2}{n}(m+2)\delta}{n-1+\frac{2\delta m}{n}}=1+\frac{\frac{4}{n}\delta}{n-1+\frac{2\delta m}{n}}.
\end{eqnarray}
If we let $\delta=\frac{1}{5m+2n}$, then $\eta=1+\frac{\frac{4}{n}\delta}{n-1+\frac{2\delta m}{n}}=1+\frac{4}{n(5m+2n)(n-1)+2m}$.

Suppose that there exists an FPTAS for finding geometric prototype with $k=2$. Let $\epsilon <\frac{4}{n(5m+2n)(n-1)+2m}$,  then we directly know that the cost of a $(1+\epsilon)$-approximation of geometric prototype is  less than $n-1+\frac{2}{n}(m+2)\delta$ if and only if $\phi$ is satisfiable. Since $\frac{1}{\epsilon}$ is polynomial on $m$ and $n$, NAE3SAT$^*$ can also be solved in polynomial time (and one of the chromatic clusters yields a valid assignment). Obviously this can only happen when P=NP.
\end{proof}

\subsection{Proof of Lemma~\ref{lem-tri}}
\label{sec-plem-tri}
Let $A=\{a_1, a_2, \cdots, a_k\}$, $B=\{b_1, b_2, \cdots, b_k\}$, and $C=\{c_1, c_2, \cdots, c_k\}$. Without loss of generality, we assume that the induced permutations of $\mathcal{M}(A, C)$ and $\mathcal{M}(C, B)$ are both $\pi(j)=j$ for $1\leq j\leq k$ (since these two permutations are independent with each other). Thus,
\begin{eqnarray}
\mathcal{M}(A, B)&\leq& \sum^k_{j=1}||a_j-b_j||^2\nonumber\\
&\leq&\sum^k_{j=1}2\big(||a_j-c_j||^2+||c_j-b_j||^2\big)\nonumber\\
&=&2\mathcal{M}(A,C)+2\mathcal{M}(C,B).
\end{eqnarray}

\subsection{Weighted Case for Core-set}
\label{sec-weight-coreset}

Here, we show the details for analyzing the core-set construction of the more general weighted case. The input is still a set of point sets $\mathbb{P}=\{P_1, P_2, \cdots, P_n\}$ with each $P_i$ containing $k$ points $\{p^i_1, p^i_2, \cdots, p^i_k\}\subset\mathbb{R}^d$; the only difference is that each $p^i_j$ has a non-negative weight $\alpha^i_j$ and $\sum^k_{j=1}\alpha^i_j=W\geq 0$. We also restrict the prototype $g(\mathbb{P})$ having $k$ weighted points with the same total weight $W$. To simplify our problem in practice, we can assume all the point weights and total weight $W$ are integers. Moreover, we replace the matching cost in Definition~\ref{def-match} by Earth Mover's Distance (EMD)~\cite{rubner2000earth}:
\begin{definition}[$\mathcal{EMD}_1(A, B)$, $\mathcal{EMD}_2(A, B)$]
\label{def-emd}
Given two point sets $A=\{a_1, a_2, \cdots, a_k\}$ and $B=\{b_1, b_2, \cdots, b_k\}$ in $\mathbb{R}^d$, where each $a_j$ (resp., $b_j$) has a weight $\alpha_j$ (resp., $\beta_j$) $\geq 0$, $\sum^k_{j=1}\alpha_j=\sum^k_{j=1}\beta_j=W$, then 
\begin{eqnarray}
\mathcal{EMD}_1(A, B)&=&\min_{F=\{f_{jl}\}}\sum^k_{j=1}\sum^k_{l=1}f_{jl}||a_j-b_l||;\\
\mathcal{EMD}_2(A, B)&=&\min_{F=\{f_{jl}\}}\sum^k_{j=1}\sum^k_{l=1}f_{jl}||a_j-b_l||^2.
\end{eqnarray}
$F=\{f_{jl}\}$ indicates a feasible flow from $A$ to $B$, i.e., each $f_{jl}\geq 0$, $\sum^k_{l=1}f_{jl}=\alpha_j$, $\sum^k_{j=1}f_{jl}=\beta_l$, and $\sum^k_{j=1}\sum^k_{l=1}f_{jl}=W$.
\end{definition}

Since we assume all the weights are integers, each point $a_j$ (resp., $b_j$) can be viewed as $\alpha_j$ (resp., $\beta_j$) overlapping points. Also, we know that an integer-valued solution exists for $F=\{f_{jl}\}$ by using linear programming~\cite{goelflow}. Thus, we directly have the following results analogous to Lemma~\ref{lem-tri} and Lemma~\ref{lem-tri2} (just replace $k$ by $W$ in the proofs).

\begin{lemma}
\label{lem-tri-weight}
Let $A$, $B$, and $C$ be three weighted $k$-point sets in $\mathbb{R}^d$, where each point has a non-negative integer weight and their total weights are equal. Then 
\begin{eqnarray}
\mathcal{EMD}_1(A, B)&\leq&\mathcal{EMD}_1(A, C)+\mathcal{EMD}_1(C, B);\nonumber\\
\mathcal{EMD}_2(A, B)&\leq&2\mathcal{EMD}_2(A, C)+2\mathcal{EMD}_2(C, B);\nonumber\\
|\mathcal{EMD}_2(A, B)-\mathcal{EMD}_2(A, C)|&\leq& (1+\frac{1}{\epsilon})\mathcal{EMD}_2(B, C)+\epsilon\mathcal{EMD}_2(A, B).
\end{eqnarray}
\end{lemma}

Based on Lemma~\ref{lem-tri-weight} and the same ideas in Section~\ref{sec-issue1} and~\ref{sec-issue2}, we can solve the two issues, bounding each sensitivity and discretizing the space of candidates for geometric prototype, for the weighted case. The only difference for issue \textbf{(\Rmnum{2})} is that we also need to consider the total $O(W^k)$ possible distributions of the total weight $W$ over the $k$ points for geometric prototype. Thus,  $\Gamma$ contains $O\big((\frac{4\sqrt{kd}}{\epsilon})^{kd}W^k\big)$ different weighted $k$-point sets in total. Consequently, the size of the sample $\mathbb{S}$ in Theorem~\ref{the-coreset} will be $O(\frac{kd}{\epsilon^2}\log\frac{kd}{\epsilon}+\frac{k}{\epsilon^2}\log W)$.

\end{document}